\documentclass[12pt,twoside]{article}
\usepackage{mattstyle}  

\makeatletter
\def\thanks#1{\protected@xdef\@thanks{\@thanks
        \protect\footnotetext{#1}}}
\makeatother

\newcommand{\possessivecite}[1]{\citeauthor{#1}'s (\citeyear{#1})}

\usepackage{pdfpages}

\usepackage{booktabs}
\usepackage{threeparttable}
\usepackage{multirow}
\usepackage{tabulary}
\newcolumntype{F}[1]{>{\raggedright\let\newline\\\arraybackslash\hspace{0pt}}m{#1}}
\newcolumntype{G}[1]{>{\centering\let\newline\\\arraybackslash\hspace{0pt}}m{#1}}
\newcolumntype{H}[1]{>{\raggedleft\let\newline\\\arraybackslash\hspace{0pt}}m{#1}}
\newcolumntype{I}{>{\centering\arraybackslash}m{3cm}}

\usepackage[section]{placeins}

\usepackage{subcaption} 

\usepackage{amsthm}
\newtheorem{proposition}{Proposition}

\graphicspath{{figures/}}

\begin{document}
\bibliographystyle{aernobold}


\begin{spacing}{0.9}
\begin{titlepage}

\title{On Track for Retirement?}
\runningheads{Matthew Olckers}{On Track for Retirement?}

\author{\large \href{https://www.matthewolckers.com/}{Matthew Olckers}\thanks{The experiment was approved by the Institutional Review Board at the Paris School of Economics (2018 017) and the survey of predictions was approved by the Monash University Human Research Ethics Committee (25207). A pre-analysis plan was registered with \href{https://aspredicted.org}{aspredicted.org}. I am grateful to Milo Bianchi, Amit Dekel, Jonathan Guryan and Birendra Rai for helpful comments.}\\ {\color{MattBlue} UNSW Sydney}
   }

\date{\small \today  }
\maketitle
\thispagestyle{empty}

\vspace{-0.3in}

\begin{abstract}
Over sixty percent of employees at a large South African company contribute the minimum rate of 7.5 percent to a retirement fund---far below the rate of 15 percent recommended by financial advisers. I use a field experiment to investigate whether providing employees with a retirement calculator, which shows projections of retirement income, leads to increases in contributions. The impact is negligible. The lack of response to the calculator suggests many employees may wish to save less than the minimum. I use a model of asymmetric information to explain why the employer sets a binding minimum.
\end{abstract}

%

\end{titlepage}
 \end{spacing}





\thispagestyle{empty}

\begin{quote}
{\it
    The Eighth Wonder of the World---is compound interest.}
\end{quote} \hspace{1.8 in} --- The Equity Savings and Loan Company, 1925

\section{Introduction}

Since as early as 1925, asset managers have extolled the wonder of compound interest.\footnote{The quote ``The Eighth Wonder of the World---is compound interest.'' is often attributed to Albert Einstein but he likely never said those words. The famous statement was used by The Equity Savings and Loan Company for an advertisement in 1925. See \href{https://quoteinvestigator.com/2019/09/09/interest/}{quoteinvestigator.com} for more details.} If you only start saving early enough, you will see your savings double---again and again. Employers and governments share a similar reverence for compound interest. Substantial tax and salary incentives encourage employees to start saving early and, in many countries, the government compels employees to save for retirement.

Employees do not seem to share the same admiration for compound interest. As the responsibility of saving for retirement has shifted from employers to employees in most countries \citep*{choi2015review}, very few employees exceed the minimum or default retirement savings rate \citep*{madrian2001default,thaler2004smart,choi2004default,chetty2014active}. Perhaps, employees do not fully appreciate the gains provided by compound interest?

People struggle to make even the most basic calculations about compound interest and inflation \citep*{lusardi2014jel} and tend to underestimate exponential growth \citep*{stango2009exponential,levy2016exponential,goda2019predict}. If employees only understood how a small increase in saving today could have a massive impact in retirement, increased retirement savings may follow.\footnote{Even though compound interest may be a basic concept, calculating optimal retirement wealth is very challenging---even for economists \citep*{skinner2007savingenough,poterba2015heterogeneity}. In the setting studied in this paper, most employees have projected retirement wealth far below common benchmarks so calculating the optimum may be less beneficial than simply appreciating the magnitude of long-term compound growth.}

In this paper, I use a field experiment at a large South African financial services company to test if making compound interest calculations easier causes increased contributions to a tax-deferred retirement account. The treatment group received a custom-built calculator to estimate income at retirement. Even though employees found the calculator helpful, the calculator had only a zero to marginally positive impact on contribution rates.

The lack of response to the calculator is surprising given that over sixty percent of the employees contribute at the minimum rate of 7.5 percent, which is expected to provide a retirement income equal to about a third of current salary. The retirement calculator suggests employees need to save at a rate of 15 percent (and sometimes more) to prevent a large drop in income at retirement. If employees were already saving at a rate close to 15 percent, we may expect little change in behavior. Since the current saving rates are so low, the lack of response is puzzling.

Beliefs about investment returns, beliefs about mortality, preferences for property, present-biased preferences or some combination of these factors may generate optimal contribution rates below 7.5 percent. Many employees may wish to contribute below the current minimum set by the company. Why would the company set a binding minimum?

I use a model of asymmetric information to explain the company's decision. The company cannot observe the employee's attitude towards retirement saving. The employee may be careful and willing to contribute regardless of the minimum or the employee may be a ``yolo" type, spending today and not worrying about tomorrow.\footnote{The term ``yolo" is derived from an acronym for ``you only live once" and refers to living in the moment and not worrying about the consequences.} The employer may set a binding minimum to prevent the yolo types from retiring without sufficient savings and damaging the reputation of the company. South Africa, like most developing countries, does not provide a state pension. The only state-provided retirement support is the needs-based Old Age Grant, worth approximately $100$ USD per month. Employees in this study earn at least $1\,200$ USD per month so they must rely on private retirement savings to sustain a similar income during retirement. In this context, employees are more likely to seek support from a previous employer than the state if they find themselves with insufficient retirement saving.

This paper contributes to a literature investigating how people respond to projections of retirement income.\footnote{ Providing income projections is not the only type of information intervention thought to impact retirement savings. Many employers, particularly in the United States, encourage retirement saving by matching the retirement contributions of their employees. The more the employee contributes, the more the employer contributes. Several field experiments focus on communicating information about the match as the arrangement provides an instant return on investment. Some find positive impacts on contribution rates \citep*{clark2014match,goldin2017match,perry2019match} while others find no statistically significant impacts \citep*{choi2011match} or even positive and negative effects depending on the wording of the communication \citep*{choi2017match}. In addition to emphasizing matching incentives, a recent field experiment emphasized other benefits, such as longevity and tax advantages, and found small positive effects but did not detect differences in effects depending on the type of information communicated \citep*{clark2019informing}.} A field experiment in the United States added retirement income and balance projections to a brochure for employees of the University of Minnesota and found small increases (0.15 percent of average salary) in retirement saving among the treatment group \citep*{goda2014projections}. The projections in the brochure were not personalized to specific employees but rather showed the marginal benefit of additional savings on projected retirement income.

Personalized information, which uses the employee's current age, account balance and savings rate, may be more useful than generic information. A field experiment in Chile used self-service terminals at government offices to provide either generic or personalized information for participants in the government's defined contribution pension plan \citep*{fuentes2018personalized}. The personalized information caused a 1.5 percent increase in the number of individuals making voluntary contributions to the pension plan and the change was concentrated among participants who overestimated their retirement income.

Besides specific interventions tested with field experiments, researchers have also used natural experiments to understand the impact of retirement income projections. People typically receive information on retirement benefits through emailed or posted statements. The administrative systems used to send out these statements can generate random variation in information provision. In the United States, the Social Security Administration's annual statement, which shows expected social security benefits, caused improvements in knowledge but no changes in behavior \citep*{mastrobuoni2011statement,carter2018statement}. In contrast, the pension statement in Germany, which shows similar information, caused saving to increase by 14 euros on average, an 11 percent increase relative to the sample mean \citep*{dolls2018statement}. A personalized statement of pension benefits in Chile increased the likelihood of voluntary retirement saving by 1.3 percent \citep*{fajnzylber2015statement}.

The natural experiments are consistent with the field experiment results of mostly small positive impacts of communicating retirement income projections. More generally, a growing base of evidence suggests that financial education initiatives improve knowledge but provide only minor impacts on behavior \citep*{fernandes2014meta,miller2015meta,kaiser2017meta}.

One exception to the trend of small impacts from financial education interventions is a field experiment involving rural farmers in China \citep*{song2019compound}. The intervention combined retirement income projections with a lesson on compound interest. The education intervention caused a 40 percent increase in savings relative to the control group---far larger than the impacts in other studies.\footnote{The timing of the intervention might explain part of the difference. \possessivecite{song2019compound} intervention was conducted at the time of enrollment and very few farmers changed their enrollment rate after the initial selection. Financial education interventions generally have larger impacts when timed to coincide with a relevant financial choice \citep*{fernandes2014meta}.}

As emphasized by \citet*{hastings2013review}, policymakers need to know which types of financial education tools are most effective to improve financial outcomes. The main innovation of the retirement calculator is allowing for interaction.\footnote{Although not focused on retirement saving, a recent experiment with recipients of conditional cash transfers in Colombia also used an interactive tool \citep*{attanasio2019tablet}. A non-government organization provided the recipients with a tablet computer containing a financial education program. The tablet includes games and other interactive elements. The researchers tested the impact of providing the tablet on a wide variety of measures, but most were not statistically different from zero after correcting for multiple hypothesis testing. One exception was a sustained increase in self-reported informal savings.} Previous experiments provided information by email, on brochures or computer display terminals. A retirement calculator allows employees to adjust the assumptions and see how different choices affect their retirement income. We may expect larger impacts from an interactive tool given that active rather than passive methods of teaching have been shown to improve learning \citep{freeman2014active}.

Policymakers may value the focus on this approach as a retirement calculator provides several advantages over benefit statements. First, a calculator provides privacy as the provider of the calculator does not need to collect any data about the employee. Second, the employee can personalize the calculation as he or she wishes, and include savings that may not be observed by a single asset manager or by the state. Even though I find no large positive impact on retirement contributions, the advantages of a reliable publicly-provided retirement calculator may warrant the cost of building and maintaining the calculator.

Besides testing a new tool, this paper contributes to the literature by studying the impact of providing retirement income projections in a new and relevant setting. As in many developing countries, the tax base in South Africa is very small and the only government-provided retirement benefit is the Old-Age Grant, which is approximately 10 percent of an entry-level salary at the company in this study. Lack of preparation for retirement will likely have much stronger impacts on welfare in this setting than in developed countries.

Finally, this paper contributes to the literature by suggesting that asymmetric information between the employer and the employees may determine the minimum contribution rate. The employer may force employees to contribute above their optimal contribution rate out of fear that certain employees will retire without sufficient savings. As the employer cannot observe which employees have sufficient savings, the employer may set a high minimum for all employees. In developing countries, such as South Africa, where high-paid employees rely on private retirement savings, such employer policies may be a vital determinant of retirement savings.

\section{Setting and Intervention}

\subsection{Retirement Saving in South Africa}

As in many countries, such as the United States and the United Kingdom, South Africa relies on a private retirement system. The South African government provides a needs-based Old Age Grant of R 1 780 per month (approximately 105 USD) but any additional retirement income must be funded by private saving.

The South African government encourages people to save for retirement using tax-deferred retirement funds. South African tax law allows people to contribute up to 27.5 percent of their pre-tax salary to a retirement fund and no income tax is paid on these contributions.\footnote{The tax-free contributions are capped at R 350 000. Visit \href{http://www.sars.gov.za/ClientSegments/Individuals/Tax-Stages/Pages/Tax-and-Retirement.aspx}{www.sars.gov.za} for more details on the tax treatment of retirement funds in South Africa.} The capital gains and dividends within the retirement fund are also not taxed.

Most employer retirement funds in South Africa use a defined contribution arrangement, as is the growing trend in many countries \citep*{choi2015review}. Retirement income is not guaranteed. The amount the employee saves and the return on the investment determines the amount the employee receives at retirement.\footnote{Despite most employers using a defined contribution arrangement, the largest employer fund---the Government Employee's Pension Fund---uses a defined benefit arrangement. This fund has 1.2 million members and over 400 000 pensioners. Government employees contribute around 20 percent of their salary and receive a guaranteed payment in retirement according to their number of years of employment. See \href{http://www.gepf.gov.za/index.php/about_us/article/who-is-gepf}{www.gepf.gov.za} for more details on the Government Employee's Pension Fund. }

Defined contribution retirement funds range in flexibility. Some employers require all employees to contribute at a fixed rate while others allow employees to choose a contribution rate within a band. Many employers, especially smaller companies, do not offer a retirement fund. Instead, these employees must use a retirement fund provided by a private asset management company.

A retirement saving system can be summarized using the analogy of the ``three-legged stool" of personal savings, government pensions and employer pensions. In South Africa, the leg of government pensions is short---only 17 percent of average earnings relative to 49 percent in OECD countries \citep*{OECD2019pensions}. Also, the South African government does not provide specific in-kind benefits to retirees. Without a large government pension, the saving choices of individuals and employers are of consequence.

In most countries, part of the population cannot rely solely on government pensions to fund retirement. Like many South Africans, these individuals must use employer pensions or personal savings to smooth consumption across retirement.


\subsection{Retirement Saving at the Company}

The company requires employees to contribute at least 7.5 percent of their pre-tax salary towards a defined contribution retirement plan. For an employee who starts working at 25 years of age and aims to retire at 65, financial advisers suggest a retirement contribution rate of 15 percent to allow the employee to draw a retirement income at a similar level to his or her pre-retirement income.

\begin{figure}[ht]
\centering
\includegraphics[width=\textwidth]{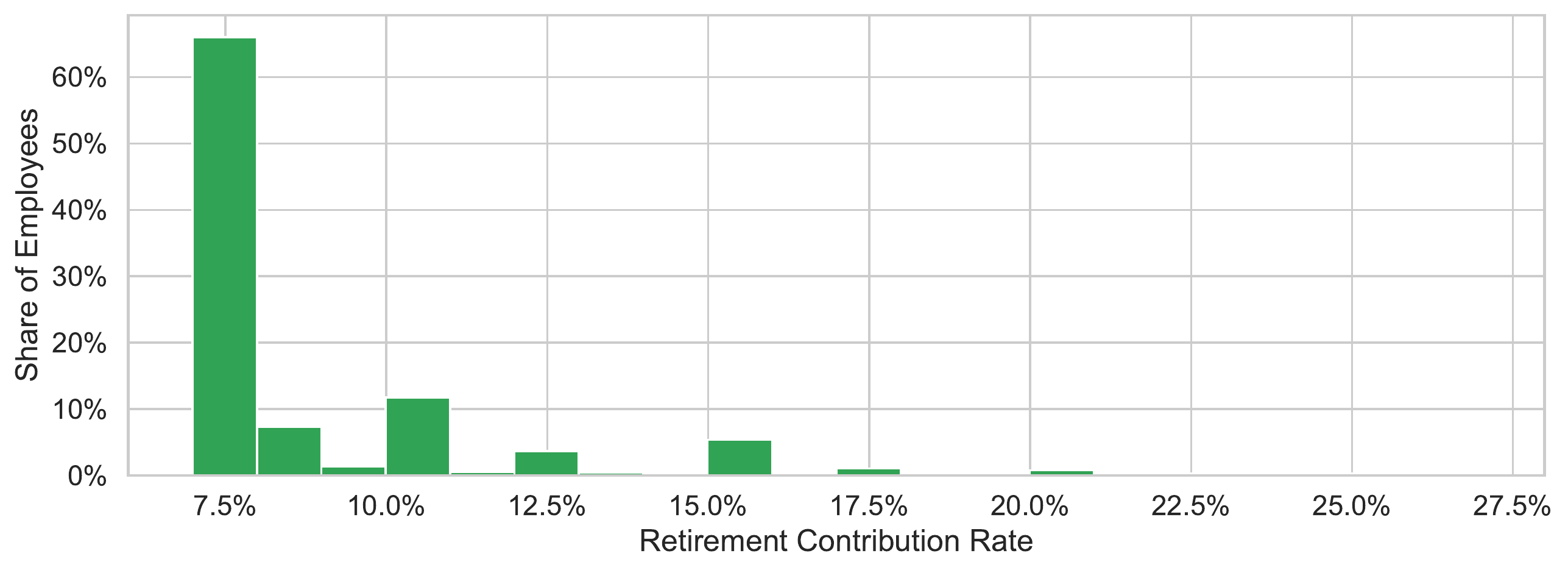}
\caption{Histogram of employee retirement contribution rates in the month before the experiment}
\label{fig:histogram}
\end{figure}

Most employees at the company contribute much less than 15 percent per month towards their retirement plan. Figure \ref{fig:histogram} shows the distribution of contribution rates chosen by employees. Over sixty percent of the employees contribute the minimum rate of 7.5 percent and only ten percent of the employees contribute 15 percent or more.

Even though legislation does not compel the company to set a minimum retirement contribution rate, this is a common practice. I sent out a survey to South African companies about their retirement policies. Of the 27 companies who responded, 52 percent impose a mandatory minimum.\footnote{See Appendix \ref{sec:example_policies} for more details on the responses.} These practices are not limited to South Africa. For example, universities in Australia compel permanent employees to contribute 17 percent of their salary to retirement when the legislation only requires 9.5 percent.

\FloatBarrier
\subsection{Retirement Calculator}

Calculating whether your retirement savings rate is sufficient to meet your retirement goals can be a difficult task. Even if you know all the inputs needed to complete the calculation, the calculation itself can be difficult. You need to use an annuity formula, which may not be obvious for many employees.

With these challenges in mind, the company developed a  retirement calculator to help employees check if they are on track with their retirement savings. The calculator has six inputs: the employee's gender, age, expected retirement age, balance of retirement savings, current salary and monthly contribution rate. Each input field provides a hint on where to find the information. For example, the field requesting the current monthly contribution rate mentions that this percentage can be found on the employee's payslip. Figure \ref{fig:calc} shows the calculator's input fields and the first results screen.

\begin{figure}[ht]
\centering
\begin{subfigure}[b]{0.45\textwidth}
\includegraphics[width=\textwidth]{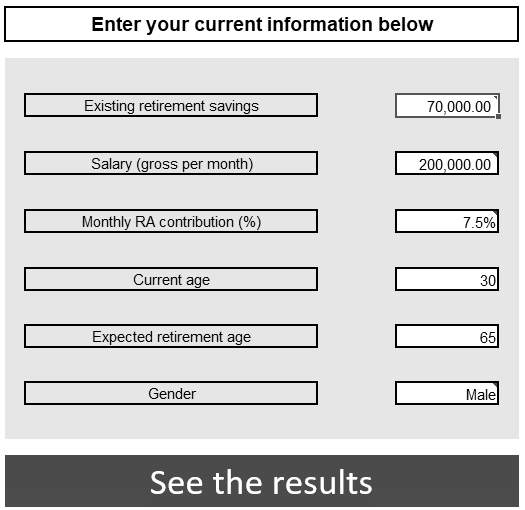}
\caption{Input fields}
\end{subfigure}
\ \
\begin{subfigure}[b]{0.45\textwidth}
\includegraphics[width=\textwidth]{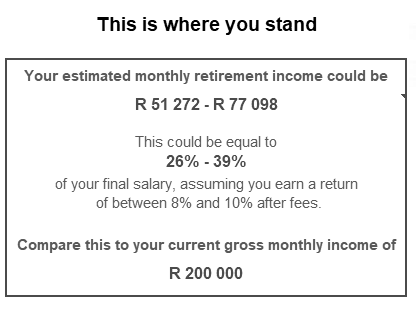}
\caption{Results screen}
\end{subfigure}
\caption{Retirement calculator}
\label{fig:calc}
\end{figure}

The result shows the estimated monthly retirement income based on the employee's inputs. For example, as shown in Figure \ref{fig:calc}, a male employee of 30 years of age with R 70 000 in retirement savings, a 7.5 percent monthly savings contribution and a salary of R 200 000 would expect to have an income of between R 51 000 and R 77 000 at retirement, which is equal to between 26 percent and 39 percent of the employee's salary.

After the first results screen, the employee is prompted to observe the impact of increasing his or her monthly contribution rate, adding a lump sum amount or a combination of the two. The results screen concludes with a link to a tax calculator for the employee to calculate his or her net salary after changes in the contribution rate and a link for the employee to process a change in his or her contribution rate.

The calculator uses several assumptions, which the employee can tweak on a separate screen. The defaults assume the employee's salary increases with inflation, the nominal investment return is between 8 and 10 percent and inflation is 5 percent. To calculate retirement income, the calculator uses recommended drawdown rates provided by the Association for Savings and Investment South Africa (ASISA), which range from 2.5 to 8 percent depending on the gender and retirement age of the employee.

\subsection{Email Emphasized a Target Retirement Income}

The calculator was sent to employees via email with the subject  ``Are you on track for 75\%?".\footnote{See appendix \ref{appendix:email} to read the body of the email.} This percentage, often called the replacement rate, refers to a target retirement income relative to the employee's salary. The 75 percent replacement rate is a common rule of thumb for adequate retirement income.\footnote{The majority of respondents to survey conducted in the United States and the Netherlands preferred a replacement of more than 75 percent \citep*{binswanger2012adequate} so 75 percent serves as a useful benchmark. The company focused on the replacement rate to give employees an easy-to-understand measure of expected retirement income.}

In the example in Figure \ref{fig:calc}, the employee reaches a replacement rate of only 39 percent. What contribution rate is needed to reach a replacement rate of 75 percent? For most employees, the contribution rate should be at least 15 percent, and sometimes far greater for older employees who have saved very little at earlier stages of their career.

The required contribution rate can be calculated using a formula derived from an annuity formula.\footnote{The derivation is available in Appendix \ref{appendix:formula}.} Let $p$ be the goal replacement rate, $r$ be the real investment return, $d$ the drawdown rate (the percentage of the employee's retirement saving he withdraws each year during retirement) and $n$ the number of years until retirement. For simplicity, assume the employee has not started saving yet.
\begin{align*}
    \text{Required contribution rate} = \dfrac{pr}{d((1+r)^n-1)}
\end{align*}
If we assume a goal replacement rate of 75 percent and a real investment return of 5 percent, an employee who starts saving at age 25 and plans to retire at 65 must save 15.52 percent of his or her salary for retirement. This common scenario explains why many retirement investment providers and financial advisers suggest a 15 percent contribution rate as a benchmark.\footnote{For example, articles from \href{https://www.fidelity.com/viewpoints/retirement/how-much-money-should-I-save}{Fidelity} in the United States, \href{https://inews.co.uk/inews-lifestyle/money/how-much-pay-pension/}{iNews} in the United Kingdom and \href{https://businesstech.co.za/news/business/241055/heres-how-much-you-should-be-saving-for-retirement/}{BusinessTech} in South Africa all suggested 15 percent as a retirement savings rate. } If the employee only starts saving at age 30, the required contribution rate jumps to 20.75 percent.

To generate a required contribution rate of 7.5 percent, the rate at which most employees contribute, we would need to make adjustments to our assumptions. We could delay retirement from 65 to 79 years of age, increase the investment return from 5 to 8 percent, increase the drawdown rate from 4 to 8 percent, or some combination of these changes. Since the required contribution rate is sensitive to the assumptions, the calculator offers the advantage of allowing the employee to personalize the inputs and assumptions.

\section{Experiment}
This section describes the design of the experiment. The experiment had two treatment arms. In the first arm, the employee received an email about the retirement calculator, and in the second, the employee received the email and a phone call.

Each year in late November the company sends letters to the employees which state the amount of the employee's performance bonus and the employee's increase in salary for the coming year. South Africa has persistent inflation of around 5 percent so salary increases at least match inflation, but can be far larger.

A few days later, the employee receives a form via email to change his or her monthly retirement contribution rate. This form is sent every year to coincide with salary increases. Even though the form is not compulsory, employees are far more likely to change their contribution rate at this time of the year. In the past three years, over 6 percent of employees change their contribution rate in December whereas less than 1 percent of employees change their contribution rate in other months of the year.

Before employees were notified of the amount of their salary increase, the team provided a randomly selected group of employees with the retirement calculator. The retirement calculator was sent to the treated employees by email. The email contained a hyperlink to the calculator.

Since employees receive many emails each day, the email about the calculator may never be read. A second treatment arm was included to mitigate this concern. In addition to the email, a team member phoned every employee included in the second treatment arm to ask for feedback on the calculator. These phone calls ensured the employee would use the calculator. The team member who conducted the phone calls worked through the list in the ``email and phone" group in random order. For the employees who did not answer, he phoned a second time. The timeline of the experiment was as follows:

{\centering
  \begin{tabulary}{\textwidth}{rL}
16 November  & Email about calculator sent to both treatment groups. \\
19-23 November  & Phone calls to the ``email and phone" group. \\
23 November  & Bonus and salary increase letters sent. \\
26-28 November & Employees receive an online form via email to change their monthly contribution rate. \\
  \end{tabulary}}

\vspace{0.75cm}

The company has over 1200 employees and I selected 775 for the experiment. I included employees who worked at the head office, were full-time permanent staff and South African citizens or permanent residents.\footnote{Temporary staff and foreign nationals are not required to enroll in the company's retirement fund.} I excluded senior management, investment analysts and any employees who knew about the experiment.

Since the participants did not know about the experiment, the experiment is a \textit{natural field experiment} in the terminology of \citet{harrison2004field}. Natural field experiments have the advantage of avoiding self-selection into the experiment and changes in behavior from being exposed to the experimental setting. Due to these benefits, a recent guide on experiments in economics placed a strong emphasis on the value of natural field experiments \citep*{czibor2019guide}.

Most of the employees are young. The average age is 33.71 years with nearly 40 percent of the employees under the age of 30. Just over 50 percent of the employees are women and 71 percent of the employees were disadvantaged by the Apartheid government.

The 775 employees included in the experiment were randomly divided into three groups: the email-only treatment group (193 employees), the email and phone call treatment group (194 employees) and the control group (388 employees). To conduct the randomization, I stratified employees into groups according to whether the employee was saving at the minimum rate, gender, age categories (-27, 28-32, 33-38, 38+ years) and whether the employee was disadvantaged by the Apartheid government. Within each stratum, half were assigned to the control group, a quarter to the email-only treatment group and a quarter to the email and phone call treatment group.\footnote{If a stratum contained a number of employees which was not a multiple of four, I pick one, two or three elements randomly (according to the number of remaining employees) from the list [Control, Control, Email Only, Email and Phone Call] without replacement.}

The stratification provides improved statistical power, especially for the analysis of heterogeneous treatment effects. For example, to check for differences in treatment effects by gender, we need a sufficient number of men and women in the treatment groups and the control groups. The stratification also helps to ensure balance between the treatment and control group on key variables.\footnote{I check for balance in Appendix \ref{appendix:balance} and find no statistically significant differences between the treatment and control groups.}

\section{Empirical Approach}

The empirical approach follows a pre-analysis plan registered at \texttt{\href{https://aspredicted.org}{aspredicted.org}}. I note any deviations from the plan as they occur.

I start by estimating intention to treat effects with the following regression equation: \vspace{-5mm}
\begin{align*}
    Contribution_i = \alpha_0 + \alpha_1 Treatment_i + s_g + \epsilon_i
\end{align*}
$Contribution_i$ measures the monthly contribution rate in December, the first month after the experiment, or the change in contribution rate from November to December.\footnote{There are two deviations from the pre-analysis plan related to the outcome $Contribution_i$. First, I did not specify using the change in contribution rate as an outcome in the pre-analysis plan. Second, I planned to use the share of the employee's bonus contributed towards the retirement fund as an outcome measure, but the company only provided an indicator for whether the employee had contributed a share of their bonus or not. Results using this indicator are available upon request.} $Treatment_i$ is a categorical variable for the two treatment arms of received an email only and received an email and a phone call. The control group is the omitted category. The $s_g$ term is a fixed effect for the stratification group.

The intention to treat effect measures the impact of receiving an email (or receiving an email and a phone call) about the retirement calculator but there is no guarantee that the employee used the calculator. Employees receive many emails every day and this email could easily have been overlooked.

The email contained a link to the calculator and I have data on which employees clicked on this link. I use two-stage least squares to estimate the local average treatment effect, which measures the average treatment effect among the employees who complied with the treatment and clicked on the link. The randomized treatment is used as an instrument for clicking the link.

\section{Results}

The empirical results suggest zero to marginally positive impacts of the retirement calculator on contribution rates. I provide an intuitive visualization of the treatment impact before showing the regression tables.

An estimation plot, shown in Figure \ref{fig:estplot_diff}, shows the distribution of the change in the contribution rate for the control group and each of the two treatment groups.\footnote{For more information on estimation plots, see \citet{ho2019estplots}.} Each point represents one employee and employees with identical values of the outcome are spread horizontally, which provides a visualization similar to a histogram. In both treatment groups, most employees made no change in their contribution rate (as we see most points in line with zero). Although the retirement calculator may have provided useful information for the employees, interacting with the calculator did not cause many employees to change their contribution rate.

\begin{figure}[ht]
\centering
\includegraphics[width=\textwidth]{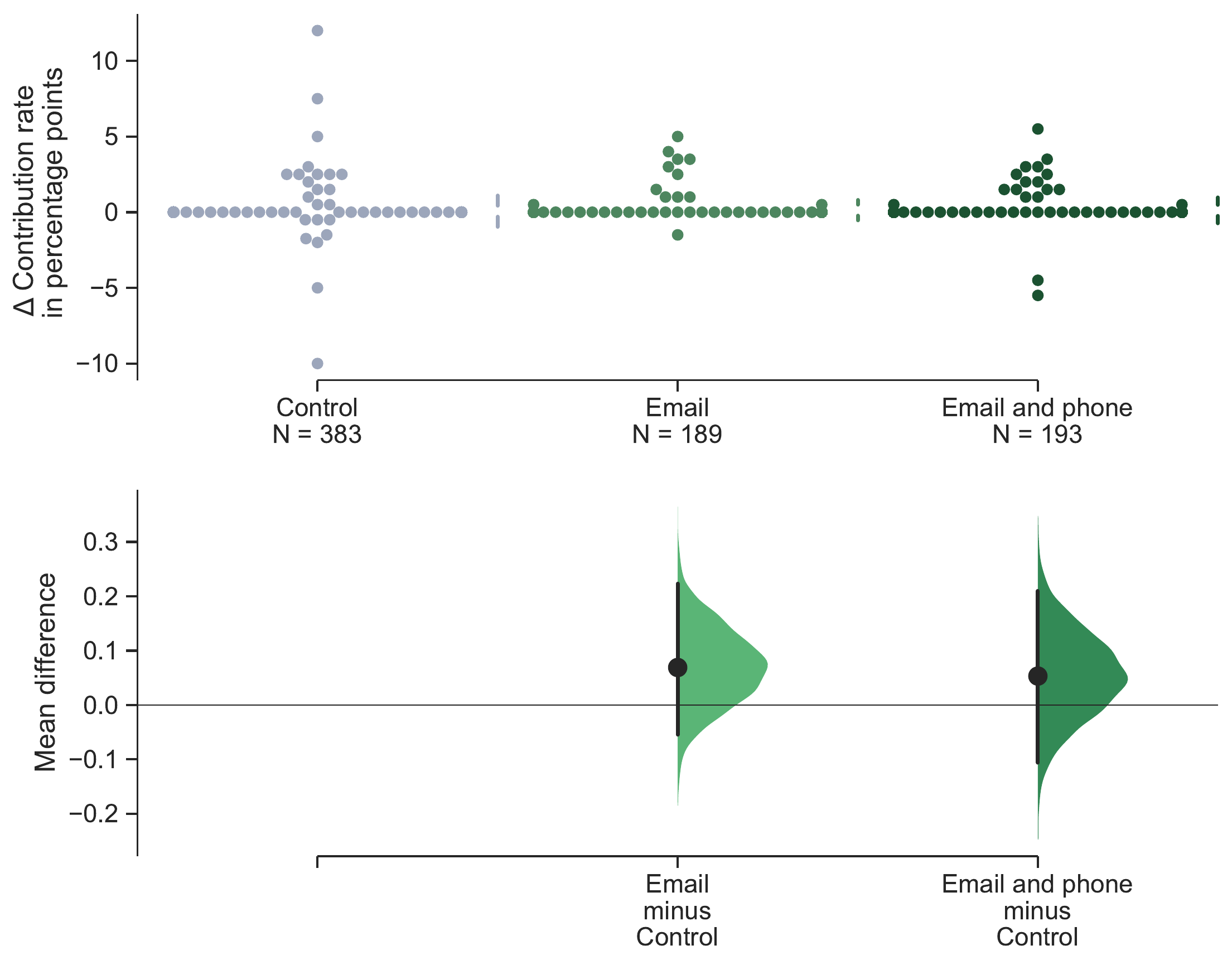}
\caption{Estimation plot for the post experiment contribution rate}
\label{fig:estplot_diff}
\end{figure}

In the bottom panel of Figure \ref{fig:estplot_diff}, the estimation plot shows a non-parametric estimation of the treatment effects. The point estimates show a less than  0.1 percentage point difference in the change in contribution rates relative to the control group. The treatment effect is bounded close to zero. The 95 percent confidence intervals, estimated via bootstrap, ranges from -0.1 to 0.2 percentage points.

I report the regression results in Table \ref{table:impact}. I show results for each of the two outcomes: the contribution rate in the first month after the experiment and the change in contribution rate before and after the experiment. For each outcome, I measure the intention to treat (ITT) and local average treatment effect (LATE).%
\footnote{Notice that there are 765 observations rather than 775. Ten employees left the company in December so we cannot observe their post-experiment contribution rate. Of the ten excluded observations, five were in the control group, four were in the email treatment group and one was in the email and phone treatment group. For historical reasons, another three employees had pre-experiment contribution rates of less than 7.5 percent (the minimum imposed by the company). At the time I selected employees to participate in the experiment, I was only shown contribution rates in categories so these three employees were shown as having a savings rate of 7.5 percent. Since contributing below the minimum is an irregular situation, I have converted the contribution rates for these three employees to 7.5 percent.}

\begin{table}[ht]
\caption{Impact of providing a retirement calculator on contribution rates}
\begin{center}
\scalebox{0.9}{
\begin{threeparttable}
\begin{tabular}{l c c c c}
\toprule
Dependent variable & \multicolumn{2}{c}{Contribution rate} & \multicolumn{2}{c}{Change in contribution rate} \\
\cmidrule(lr){2-3} \cmidrule(lr){4-5}
 & ITT & LATE & ITT & LATE \\
\midrule
Email only      & $0.079$            &                    & $0.070$            &                    \\
                & $ [-0.345; 0.504]$ &                    & $ [-0.073; 0.212]$ &                    \\
Email and phone & $0.198$            &                    & $0.052$            &                    \\
                & $ [-0.245; 0.641]$ &                    & $ [-0.107; 0.212]$ &                    \\
Clicked on link &                    & $0.304$            &                    & $0.086$            \\
                &                    & $ [-0.368; 0.975]$ &                    & $ [-0.159; 0.331]$ \\
\midrule
Control mean    & $9.010$            & $9.010$            & $0.066$            & $0.066$            \\
R$^2$           & $0.471$            & $0.473$            & $0.035$            & $0.037$            \\
Num. obs.       & $765$              & $765$              & $765$              & $765$              \\
RMSE            & $2.484$            & $2.475$            & $0.916$            & $0.914$            \\
\bottomrule
\end{tabular}
\begin{tablenotes}[flushleft]
\small{\item All specifications use strata fixed effects and robust standard errors. ITT stands for Intention To Treat and LATE stands for Local Average Treatment Effect. The LATE estimate uses the treatment as an instrument to estimate the impact for those employees who click on the link to the retirement calculator.}
\end{tablenotes}
\end{threeparttable}
}
\label{table:impact}
\end{center}
\end{table}

The LATE estimate uses the treatment as an instrument for having clicked on the link to the calculator. In the email group, 27 percent of employees clicked on the link to the calculator, and in the email and phone group, 65 percent clicked on the link to the calculator.\footnote{The link within the email is unique to the recipient so the software records the interaction whether the recipient reads email on a mobile device or a computer. The calculator did not record additional usage statistics. We cannot measure how each employee engaged with the calculator.} We assume no employees in the control group opened the email since the email was not sent to the control group.

The treatment effects are small in all specifications. We cannot reject the hypothesis that the effects are zero at the conventional level of statistical significance. The table shows the 95 percent confidence intervals below each specification. From these figures, we can conclude that provision of the calculator caused a less than one percentage point increase in contribution rates.

To ensure the low average treatment effects do not conceal large impacts for certain subgroups, in Appendix \ref{sec:het}, I test for heterogeneous treatment effects using groups specified in a pre-analysis plan. I do not find noticeable differences by age, gender or Apartheid era racial classifications. I also do not find evidence for heterogeneous treatment effects using the causal forests algorithm of \citet*{wager2018forest}.

One limitation is that I do not observe saving in private retirement accounts. However, it is unlikely that employees changed retirement savings through other means than the company retirement account. Since contributions to the retirement account are not taxed and only the company retirement account can be included in the monthly payroll, the employees have a strong incentive to use the company retirement account to save for retirement. The retirement account also offers a wide range of mutual funds so it is unlikely employees would seek to open a retirement account with an external provider to access a specific investment portfolio.

I assumed employees in the control group did not use the calculator. Perhaps this assumption is false. For example, employees in the treatment group could simply forward the hyperlink to the calculator to their colleagues.\footnote{Information passed on by peers can impact behavior. In an early field experiment, employees were more likely to increase their retirement contributions if their peers in the same department attended an employee benefits fair \citep*{duflo2003information}. In a more recent experiment, employees who received information about their peer's contribution rates decreased their contribution rates \citep*{beshears2015peer}.} This possible contamination of the control group could bias our estimates of the treatment effect downwards.

Possible contamination of the control group is not a major concern for three reasons. First, the share of employees who changed their contribution rate after the experiment was lower than at the same point in previous years. Just 5.9 percent of employees changed their contribution rate at the time of the experiment in comparison to 9.5 percent the year before and 8.1 percent two years before. If contamination of the control group is masking large treatment impacts we would expect the share of employees who change their contribution rate to rise, not fall.


Second, the treatment groups did not know they were part of an experiment. Employees often receive emails sent to all employees in the company and the email sent with the calculator was very similar. There was no reason for the treatment group to assume that some of their colleagues did not receive the same email. This was also the first experiment the company conducted with the employees.

Third, we can recalculate the local average treatment effect for the treated employees.\footnote{I did not specify this approach in the pre-analysis plan.} Within this group, we know exactly who clicked on the link to the calculator and who did not. Therefore, even if the control group was contaminated, we can check that the contamination is not concealing large average treatment effects.

In the email and phone treatment group, 65 percent of employees clicked on the link to the calculator while only 27 percent clicked on the link in the email treatment. I run the same specification as in the second column of Table \ref{table:impact} on the treated employees and the point estimate is -0.012 percentage points with a 95 percent confidence interval of between -0.391 and 0.368 percentage points. Since we do not find a large effect in this group, where we know exactly who used the calculator, we can dismiss the concerns that the possible contamination of the control group is concealing large average treatment effects.

In summary, the estimates are precise enough to conclude that the calculator did not cause an economically significant change in behavior. Given that most employees currently save at the minimum contribution rate, there was scope for large increases in contribution rates and yet even the upper limit of the confidence intervals are very small.

To provide evidence that these small treatment effects are surprising, I asked researchers and students to predict the impact of providing the retirement calculator on contribution rates. Using an online survey, I explained the context of the experiment and I provided the average contribution rates in the control and each of the treatment groups in the month before the experiment. I asked the respondents to guess the average contribution rate after the provision of the retirement calculator.\footnote{Collecting predictions of the treatment impact protects from hindsight bias. See \citet{dellavigna2019predict} for more information on this approach.} The respondent with the most accurate guess received a 50 AUD Amazon voucher.

\begin{figure}
    \centering
    \includegraphics[width=\textwidth]{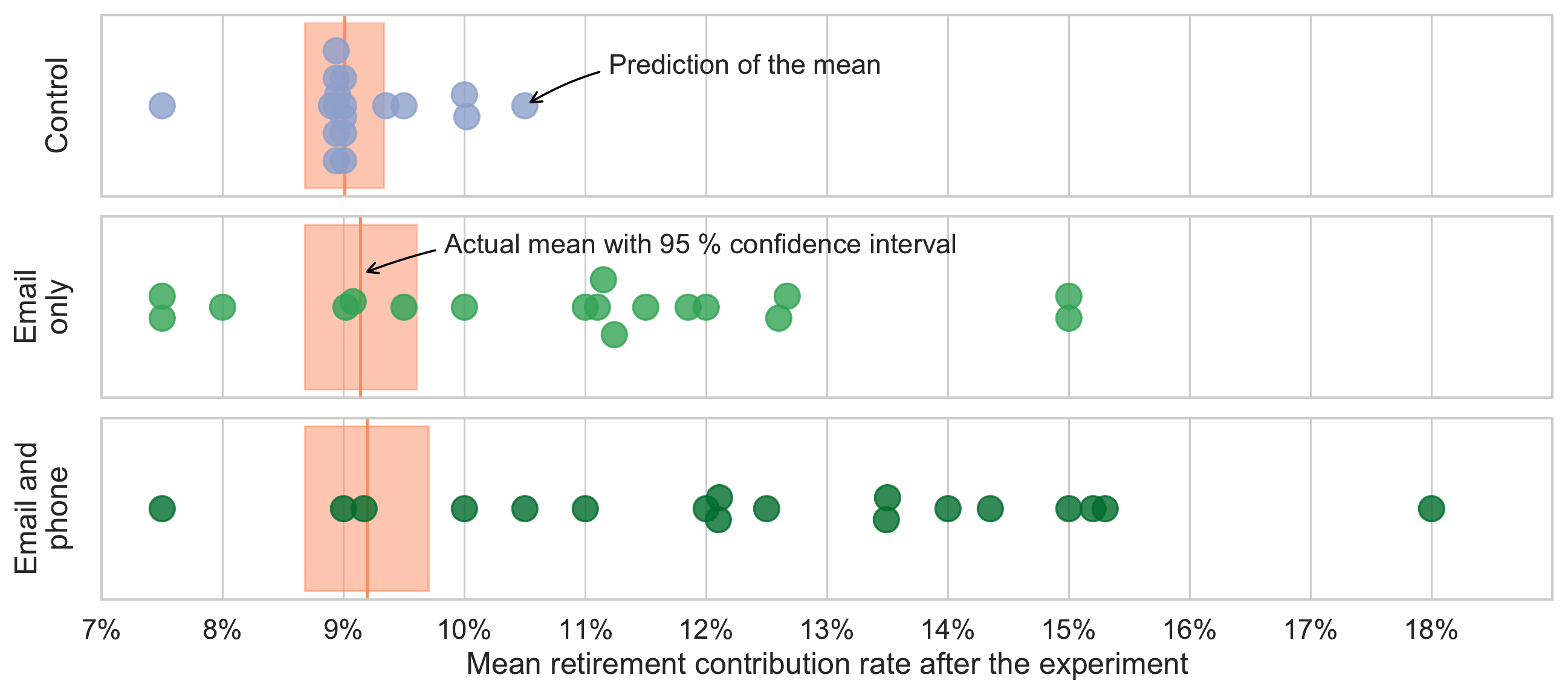}
    \caption{Predictions of the post experiment contribution rates}
    \label{fig:prediction}
\end{figure}

Figure \ref{fig:prediction} shows the predictions of 18 respondents from Monash University (8 economics students, 9 academic economists and 1 unidentified).\footnote{A recent analysis of prediction methods recommends that 15 to 30 respondents are sufficient to gain a precise estimate of the average forecast \citep*{dellavigna2020}} The true average contribution rate is shown by an orange line and the shaded region shows the 95 percent confidence interval. Most respondents guessed the control group average correctly but overestimated the average in each of the treatment groups. Only two of the respondents, predicted close to zero treatment effects in both of the treatment groups. Most respondents anticipated a larger treatment effect in the email and phone group than the email-only group.

The predictions provide evidence that the results of this experiment were unexpected. In the next section, I propose an explanation for the current equilibrium.

\section{Explaining the Current Equilibrium}

\subsection{Why Employees May Want to Save Less than the Minimum}

There are many reasons why contribution rates of 7.5 percent or lower may be optimal for employees. Examples include optimistic beliefs of investment returns, pessimistic beliefs of mortality, a strong preference for property or present-biased preferences.

The recommended contribution rate of 15 percent rests on assumptions about investment returns, retirement age and drawdown rates during retirement. Adjusting these assumptions can have a large impact on the optimal contribution rate. For example, some retirement funds have provided annual real investment returns of over 9 percent over the past 20 years---much higher than the 5 percent real return assumed by the retirement calculator. An investment return of 9 percent would decrease the recommended contribution rate from 15 percent to 5.5 percent for an employee who starts contribution at 25 and plans to retire at 65 years of age.\footnote{In Appendix \ref{appendix:formula}, I discuss the formula used in these calculations.}

The recommended drawdown rates may be too conservative for some employees. A higher drawdown rate provides more retirement income in each year of retirement and, therefore, reduces the contribution rate required to achieve a given level of retirement income. A recent survey finds that younger people underestimate survival and these beliefs cause lower levels of saving \citep*{heimer2019yolo}. South Africa has an average life expectancy of 64.1 years, which is below the normal retirement age of 65 years.\footnote{Life expectancy data sourced from \url{https://ourworldindata.org/life-expectancy}.} Incorporating this life expectancy would lead to a much lower optimal contribution rate than a life expectancy of 80 or 90 years.

An employee may have a strong preference to own property and focus on repaying a home loan. Even though retirement funds provide generous tax incentives, if a home loan takes up the bulk of the employee's budget, he may have a high propensity to consume from his remaining budget.

The experiment tackled exponential growth bias. Other behavioral biases besides exponential growth bias, such as present-based preferences \citep*{laibson1998lifecycle,goda2019predict} can generate low contribution rates.\footnote{If the employees do have present-biased preferences, the utility cost of increasing their contribution rate the following month may be too high. Even though the intervention coincided with a salary increase, the employees may be resistant to any immediate reduction in their salary. Perhaps the provision of the retirement calculator should coincide with offering a \textit{Save More Tomorrow} plan \citep*{thaler2004smart}. Instead of offering an immediate increase in contribution rates, the company could offer a delayed and automated schedule of increases for the employee to commit to in advance.}

Even without present bias, an employee's optimal contribution rate may vary over their lifetime and move below the minimum of 7.5 percent at some point. The seminal theories of life-cycle consumption show that an employee with rising income may wish to only start saving for retirement many years after they have started working \citep*{modigliani1954utility,friedman1957permanent}. Data on lifetime savings behavior indicates that people do optimize retirement saving and optimal savings rates vary over people's lifetimes \citep*{scholz2006savings}.

Any combination of the above reasons could cause optimal contribution rates to be below the minimum of 7.5 percent. In the next section, I explore why the employer sets a binding minimum.

\subsection{Why the Employer Sets a Binding Minimum}

Assume that some employees have an optimal retirement contribution rate of below 7.5 percent. Why does the employer set a binding minimum? I introduce a model of asymmetric information to explain the employer's choice.

Assume there are two types of employees, careful and yolo (an acronym for ``you only live once"). The careful employee will build up retirement savings regardless of the employer's policies. The yolo employee likes to live in the moment and would rather consume than focus on saving for retirement.\footnote{A survey in urban areas of South Africa found that one in three employees have no formal retirement saving \citep*{oldmutual2018}.}

In order to attract talented job applicants and gain customer loyalty, the employer values the company's reputation \citep*{eccles2007reputation}. The satisfaction of current and past employees will be a strong determinant of the company's reputation \citep*{cravens2006employees}. If the employer is faced with employees with present-biased preferences or self-control problems, such as the yolo types, the employer may impose restrictions to preserve the company's reputation.\footnote{A global survey of company executives reports that executives ``rate reputation risk as more important or much more important than other strategic risks their companies are facing. In addition, 88 percent say their companies are explicitly focusing on managing reputation risk." \citep*{deloitte2014reputation} } \citet{laibson2018private} calls these restrictions \textit{private paternalism}.

Retirement savings evolve over long time horizons so typically the state has a stronger incentive to enforce paternalistic retirement saving policies than the employer. However, in a setting where the state does have the capacity to provide retirement benefits, the employer may take on these paternalistic responsibilities to preserve the company's reputation. An employee who retires from the company without sufficient retirement savings may tarnish the company's reputation.\footnote{As \citet{bodie1989pensions} explains, ``While it is certainly true that employers and employees have conflicting economic interests, in many respects their interests coincide. Employers who acquire a reputation for taking care of the retirement needs of their employees may find it easier to recruit and retain higher quality employees in the future."}

We can represent the interaction between employer and the employee by the following game tree.
\begin{center}
    \includegraphics[width=0.75\textwidth]{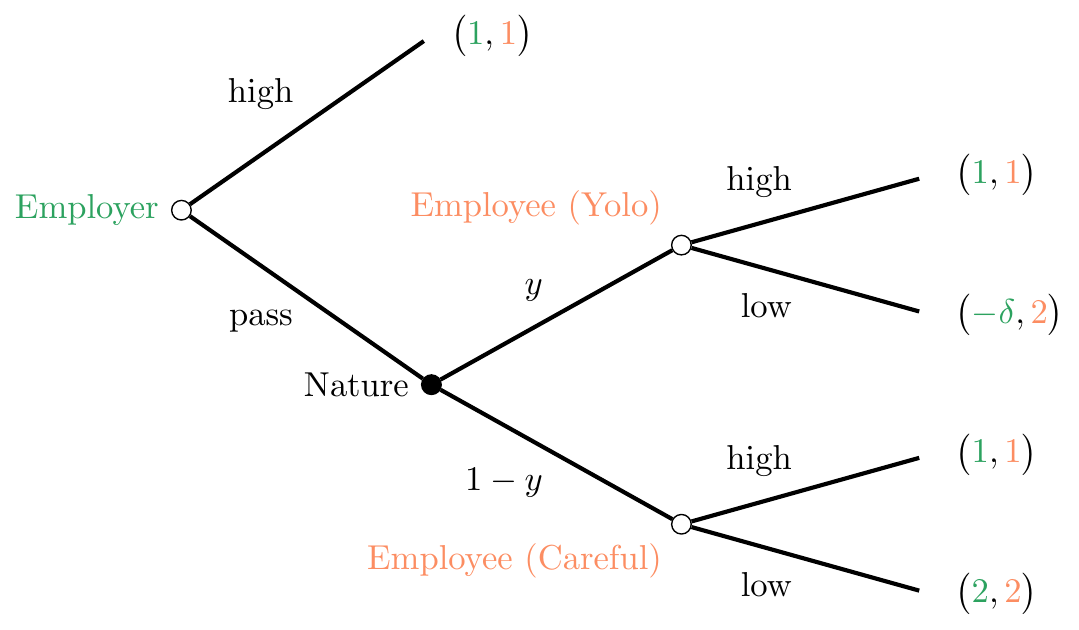}
\end{center}

The employer starts by deciding if a high minimum should be imposed or if the decision on the minimum is passed on to the employee. Think of this choice as a contract negotiation. The employer can choose to set a high minimum retirement contribution rate or give the employee the flexibility to choose the minimum (low or high). If the employer sets a high minimum, the game ends and both the employer and employee receive a payoff of 1.

If the employer allows flexibility, the employer will face a yolo employee with probability $y$ and a careful employee with probability $1-y$. Both types of employees prefer a low minimum to a high minimum, but for different reasons. The careful employee contributes above the high minimum but he likes the additional flexibility of the low minimum---in case he is hit with a temporary income shock. The yolo employee prefers the low minimum because he wants as much disposable income as possible.

The employer prefers a low minimum if faced with a careful employee and a high minimum if faced with a yolo employee. If the employer allows a low minimum and faces a yolo employee, the employer receives a payoff of $-\delta$, where $\delta$ represents the magnitude of the damage the yolo employee inflicts to the company's reputation.

The subgame perfect equilibrium of this game will depend on $\delta$ and $y$. The higher $\delta$ the more damage a yolo employee could cause and the employer is more likely to set a high minimum. The higher $y$ the more likely the employer will face a yolo employee and the employer is more likely to set a high minimum. I summarize the result in the following proposition.
\begin{proposition}
The subgame perfect equilibrium of the game between the employer and employee is
\vspace{-0.8cm}
\begin{align*}
    \{\text{high} \ ; (\text{low}, \text{low})\} & \text{ \ if \ } \delta + 2 > \dfrac{1}{y} \\
    \{\text{pass} \ ; (\text{low}, \text{low})\} & \text{ \ if \ } \delta + 2 < \dfrac{1}{y} \\
     \{q * \text{high} \otimes (1-q)* \text{pass} \ ; (\text{low}, \text{low})\} & \text{ \ if \ } \delta + 2 = \dfrac{1}{y} \text{ \ for \ } q \in [0,1] .
\end{align*}
where $\otimes$ represents a mixed strategy.
\end{proposition}

\begin{proof}
For the employer, the payoff of the high action is $1$ and the expected payoff of the pass action is $y(-\delta) + (1-y)2$. Therefore the employer prefers the high action if $ y(-\delta) + (1-y)2 > 1 $, which can be simplified to $\delta + 2 > \dfrac{1}{y}$. The employee has a dominant strategy to choose the low action. Therefore, all of the equilibria listed in the proposition are Nash Equilibria. Besides the full game itself, the only subgame stems from the Nature node. In this subgame, the employee has a dominant strategy to play (Low, Low). Therefore, all of the Nash equilibria are subgame perfect.
\end{proof}

The probability of facing a yolo employee, $y$, and the magnitude of the damage this type of employee could cause, $\delta$, may depend on the characteristics of the employer. For example, an employer with a strong brand may incur greater damage from a yolo employee than an employer with an unknown brand. The damage may also depend on the general retirement system. Due to the relatively small old age grants in developing countries, a yolo employee is more likely to retire without sufficient savings.

The model assumes that the employee's retirement saving decisions impact the employer's utility through reputation concerns. Other channels could generate this impact. The employer's managers could have paternalistic preferences and derive utility from their own aspirations for retirement \citep*{ambuehl2021motivates}. Government rules, such as the laws in the United States preventing favoritism of certain employees, could also drive this impact. I leave further investigation of employer preferences to future research.

\section{Conclusion}

This setting was ideal for detecting large responses to a retirement calculator showing income projections. Most employees currently contribute at the minimum rate of 7.5 percent, which is far below the recommended rate of 15 percent. The company provided the calculator one week before employees were notified of their annual salary increase and performance bonus. Any changes in contribution rate coincided with the first month of the salary increase so employees could raise their contribution rates without reducing their nominal salary.

Despite the potential for large responses, the retirement calculator caused a negligible change in contribution rates. This experiment adds to a growing body of evidence that financial education initiatives do little to change financial behavior \citep*{fernandes2014meta,miller2015meta,kaiser2017meta}.

Many reasons may explain why many employees may wish to contribute less than 7.5 percent of their salary to a retirement fund. I discussed beliefs about investment returns, beliefs about mortality, preferences for property and present-biased preferences.

I use a model of asymmetric information to explain why the employer may set a binding minimum contribution rate. The employer worries that some of the employees may not save enough for retirement and then damage the reputation of the company once they retire. Since the employer cannot observe the employee's attitude towards retirement saving, the employer may prefer to force all employees to save a minimum amount.

{\small
\bibliography{references}

@unpublished{attanasio2019tablet,
  title={Freeing Financial Education via Tablets: Experimental Evidence from Colombia},
  author={Attanasio, Orazio and Bird, Matthew and Cardona-Sosa, Lina and Lavado, Pablo},
  year={2019},
  note={NBER Working Paper 25929}
}

@article{thaler2004smart,
    author = {Thaler, Richard H. and Benartzi, Shlomo},
    title = {Save More Tomorrow™: Using Behavioral Economics to Increase Employee Saving},
    journal = {Journal of Political Economy},
    volume = {112},
    number = {S1},
    pages = {S164-S187},
    year = {2004}
}

@article{beshears2015peer,
    author = {Beshears, John and Choi, James J. and Laibson, David and Madrian, Brigitte C. and Milkman, Katherine L.},
    title = {The Effect of Providing Peer Information on Retirement Savings Decisions},
    journal = {The Journal of Finance},
    volume = {70},
    number = {3},
    pages = {1161--1201},
    year = {2015}
}

@article{binswanger2012adequate,
  title={What is an Adequate Standard of Living During Retirement?},
  author={Binswanger, Johannes and Schunk, Daniel},
  journal={Journal of Pension Economics \& Finance},
  volume={11},
  number={2},
  pages={203--222},
  year={2012},
  publisher={Cambridge University Press}
}

@article{carter2018statement,
  title={Can Information Change Personal Retirement Savings? Evidence from Social Security Benefits Statement Mailings},
  author={Carter, Susan Payne and Skimmyhorn, William},
  journal={AEA Papers and Proceedings},
  volume={108},
  pages={93--97},
  year={2018}
}

@article{chetty2014active,
  title={Active vs. Passive Decisions and Crowd-out in Retirement Savings Accounts: Evidence from Denmark},
  author={Chetty, Raj and Friedman, John N and Leth-Petersen, S{\o}ren and Nielsen, Torben Heien and Olsen, Tore},
  journal={The Quarterly Journal of Economics},
  volume={129},
  number={3},
  pages={1141--1219},
  year={2014},
  publisher={MIT Press}
}

@incollection{choi2004default,
  title={For Better or For Worse: Default Effects and 401 (k) Savings Behavior},
  author={Choi, James J and Laibson, David and Madrian, Brigitte C and Metrick, Andrew},
  booktitle={Perspectives on the Economics of Aging},
  pages={81--126},
  year={2004},
  publisher={University of Chicago Press}
}

@article{choi2011match,
  title={\$100 Bills on the Sidewalk: Suboptimal Investment in 401 (k) plans},
  author={Choi, James J and Laibson, David and Madrian, Brigitte C},
  journal={Review of Economics and Statistics},
  volume={93},
  number={3},
  pages={748--763},
  year={2011},
  publisher={MIT Press}
}

@article{choi2015review,
  title={Contributions to Defined Contribution Pension Plans},
  author={Choi, James J},
  journal = {Annual Review of Financial Economics},
  volume = {7},
  number = {1},
  pages = {161--178},
  year = {2015},
  doi = {10.1146/annurev-financial-111914-041834}
}

@article{choi2017match,
  title={Small Cues Change Savings Choices},
  author={Choi, James J and Haisley, Emily and Kurkoski, Jennifer and Massey, Cade},
  journal={Journal of Economic Behavior \& Organization},
  volume={142},
  pages={378--395},
  year={2017},
  publisher={Elsevier}
}

@article{clark2014match,
  title={Can Simple Informational Nudges Increase Employee Participation in a 401 (k) Plan?},
  author={Clark, Robert L and Maki, Jennifer A and Morrill, Melinda Sandler},
  journal={Southern Economic Journal},
  volume={80},
  number={3},
  pages={677--701},
  year={2014},
  publisher={Wiley Online Library}
}

@article{clark2019informing,
  title={Informing Retirement Savings Decisions: A Field Experiment on Supplemental Plans},
  author={Clark, Robert L and Hammond, Robert G and Morrill, Melinda S and Khalaf, Christelle},
  journal={Economic Inquiry},
  volume={57},
  number={1},
  pages={188--205},
  year={2019},
  publisher={Wiley Online Library}
}

@article{czibor2019guide,
  title={The Dozen Things Experimental Economists Should Do (More of)},
  author={Czibor, Eszter and Jimenez-Gomez, David and List, John A},
  journal={Southern Economic Journal},
  volume={86},
  number={2},
  pages={371--432},
  year={2019}
}

@article{dellavigna2019predict,
  title={Predict Science to Improve Science},
  author={DellaVigna, Stefano and Pope, Devin and Vivalt, Eva},
  journal={Science},
  volume={366},
  number={6464},
  pages={428--429},
  year={2019},
  publisher={American Association for the Advancement of Science}
}

@article{dellavigna2020,
Author = {DellaVigna, Stefano and Otis, Nicholas and Vivalt, Eva},
Title = {Forecasting the Results of Experiments: Piloting an Elicitation Strategy},
Journal = {AEA Papers and Proceedings},
Volume = {110},
Year = {2020},
Month = {May},
Pages = {75-79},
DOI = {10.1257/pandp.20201080}}

@article{dolls2018statement,
  title={Do Retirement Savings Increase in Response to Information about Retirement and Expected Pensions?},
  author={Dolls, Mathias and Doerrenberg, Philipp and Peichl, Andreas and Stichnoth, Holger},
  journal={Journal of Public Economics},
  volume={158},
  pages={168--179},
  year={2018},
  publisher={Elsevier}
}

@article{duflo2003information,
    author = {Duflo, Esther and Saez, Emmanuel},
    title = "{The Role of Information and Social Interactions in Retirement Plan Decisions: Evidence from a Randomized Experiment}",
    journal = {The Quarterly Journal of Economics},
    volume = {118},
    number = {3},
    pages = {815--842},
    year = {2003}
}

@article{fajnzylber2015statement,
  title={Knowledge, Information, and Retirement Saving Decisions: Evidence from a Large-Scale Intervention in Chile},
  author={Fajnzylber, Eduardo and Reyes, Gonzalo},
  journal={Economia},
  volume={15},
  number={2},
  pages={83--117},
  year={2015},
  publisher={JSTOR}
}

@article{fernandes2014meta,
  title={Financial Literacy, Financial Education, and Downstream Financial Behaviors},
  author={Fernandes, Daniel and Lynch, John G Jr and Netemeyer, Richard G},
  journal={Management Science},
  volume={60},
  number={8},
  pages={1861--1883},
  year={2014},
  publisher={INFORMS}
}

@article{freeman2014active,
  title={Active Learning Increases Student Performance in Science, Engineering, and Mathematics},
  author={Freeman, Scott and Eddy, Sarah L and McDonough, Miles and Smith, Michelle K and Okoroafor, Nnadozie and Jordt, Hannah and Wenderoth, Mary Pat},
  journal={Proceedings of the National Academy of Sciences},
  volume={111},
  number={23},
  pages={8410--8415},
  year={2014},
  publisher={National Acad Sciences}
}

@incollection{friedman1957permanent,
 chapter = 3,
 pages = {20-37},
 title = {The Permanent Income Hypothesis},
 author = {Friedman, Milton},
 booktitle = {A Theory of the Consumption Function},
 publisher = {Princeton University Press},
 year = {1957}
}

@unpublished{fuentes2018personalized,
  title={Personalized Information as a Tool to Improve Pension Savings: Results from a Randomized Control Trial in Chile},
  author={Fuentes, Olga and Lafortune, Jeanne and Riutort, Julio and Tessada, Jos{\'e} and Villatoro, F{\'e}lix},
  note={Working Paper},
  year={2018}
}

@article{goda2014projections,
  title={What Will My Account Really Be Worth? Experimental Evidence on How Retirement Income Projections Affect Saving},
  author={Goda, Gopi Shah and Manchester, Colleen Flaherty and Sojourner, Aaron J},
  journal={Journal of Public Economics},
  volume={119},
  pages={80--92},
  year={2014},
  publisher={Elsevier}
}

@article{goda2019predict,
  author = {Goda, Gopi Shah and Levy, Matthew and Manchester, Colleen Flaherty and Sojourner, Aaron and Tasoff, Joshua},
  title = {Predicting Retirement Saving using Survey Measures of Exponential-Growth Bias and Present Bias},
  journal = {Economic Inquiry},
  year = {2019},
  pages = {1636--1658},
  volume = {57},
  number = {3}
}

@article{goldin2017match,
  title={Retirement Contribution Rate Nudges and Plan Participation: Evidence from a Field Experiment},
  author={Goldin, Jacob and Homonoff, Tatiana and Tucker-Ray, Will},
  journal={AEA Papers and Proceedings},
  volume={107},
  pages={456--61},
  year={2017}
}

@article{harrison2004field,
  title={Field experiments},
  author={Harrison, Glenn W and List, John A},
  journal={Journal of Economic Literature},
  volume={42},
  number={4},
  pages={1009--1055},
  year={2004}
}

@article{hastings2013review,
  title={Financial Literacy, Financial Education, and Economic Outcomes},
  author={Hastings, Justine S and Madrian, Brigitte C and Skimmyhorn, William L},
  journal={Annual Review of Economics},
  volume={5},
  number={1},
  pages={347--373},
  year={2013},
  publisher={Annual Reviews}
}

@article{heimer2019yolo,
  title={YOLO: Mortality Beliefs and Household Finance Puzzles},
  author={Heimer, Rawley and Myrseth, Kristian Ove R and Schoenle, Raphael},
  journal={Journal of Finance},
  year={2019}
}

@article{ho2019estplots,
  title={Moving Beyond P-Values: Data Analysis with Estimation Graphics},
  author={Ho, Joses and Tumkaya, Tayfun and Aryal, Sameer and Choi, Hyungwon and Claridge-Chang, Adam},
  journal={Nature Methods},
  pages={565--566},
  volume={16},
  number={7},
  year={2019},
  publisher={Nature Publishing Group}
}

@article{kaiser2017meta,
  title={Does Financial Education Impact Financial Literacy and Financial Behavior, and If So, When?},
  author={Kaiser, Tim and Menkhoff, Lukas},
  journal={World Bank Economic Review},
  volume={31},
  number={3},
  pages={611--630},
  year={2017},
  publisher={World Bank Group}
}

@article{laibson1998lifecycle,
  title={Life-Cycle Consumption and Hyperbolic Discount Functions},
  author={Laibson, David},
  journal={European Economic Review},
  volume={42},
  number={3-5},
  pages={861--871},
  year={1998},
  publisher={Elsevier}
}

@article{levy2016exponential,
  title={Exponential-Growth Bias and Lifecycle Consumption},
  author={Levy, Matthew and Tasoff, Joshua},
  journal={Journal of the European Economic Association},
  volume={14},
  number={3},
  pages={545--583},
  year={2016},
  publisher={Oxford University Press}
}

@article{lusardi2014jel,
  title={The Economic Importance of Financial Literacy: Theory and Evidence},
  author={Lusardi, Annamaria and Mitchell, Olivia S},
  journal={Journal of Economic Literature},
  volume={52},
  number={1},
  pages={5--44},
  year={2014}
}

@article{madrian2001default,
    author = {Madrian, Brigitte C. and Shea, Dennis F.},
    title = "{The Power of Suggestion: Inertia in 401(k) Participation and Savings Behavior}",
    journal = {The Quarterly Journal of Economics},
    volume = {116},
    number = {4},
    pages = {1149--1187},
    year = {2001}
}

@article{mastrobuoni2011statement,
  title={The Role of Information for Retirement Behavior: Evidence Based on the Stepwise Introduction of the Social Security Statement},
  author={Mastrobuoni, Giovanni},
  journal={Journal of Public Economics},
  volume={95},
  number={7-8},
  pages={913--925},
  year={2011},
  publisher={Elsevier}
}

@article{miller2015meta,
  title={Can You Help Someone Become Financially Capable? A Meta-Analysis of the Literature},
  author={Miller, Margaret and Reichelstein, Julia and Salas, Christian and Zia, Bilal},
  journal={The World Bank Research Observer},
  volume={30},
  number={2},
  pages={220--246},
  year={2015},
  publisher={Oxford University Press}
}

@incollection{modigliani1954utility,
  title={Utility analysis and the consumption function: An interpretation of cross-section data},
  author={Modigliani, Franco and Brumberg, Richard},
  booktitle={Post-keynesian economics},
  editor={Kurihara, Kenneth K.},
  volume={1},
  pages={388--436},
  year={1954},
  publisher={New Brunswick, NJ}
}

@unpublished{perry2019match,
  title={Everyone Else Is Making a Mistake: Effects of Peer Error on Saving Decisions},
  author={Perry, Elizabeth},
  note={SSRN Working Paper 3348672},
  year={2019}
}

@article{poterba2015heterogeneity,
  Author = {Poterba, James M},
  Title = {Reflections of the Holland Medal Recipient: Saver Heterogeneity and the Challenge of Assessing Retirement Saving Adequacy},
  Journal = {National Tax Journal},
  Volume = {68},
  Number = {2},
  Year = {2015},
  Pages = {377--388},
  DOI = {10.17310/ntj.2015.2.06}
}

@article{scholz2006savings,
  title={Are Americans Saving “Pptimally” for Retirement?},
  author={Scholz, John Karl and Seshadri, Ananth and Khitatrakun, Surachai},
  journal={Journal of Political Economy},
  volume={114},
  number={4},
  pages={607--643},
  year={2006},
  publisher={The University of Chicago Press}
}

@article{skinner2007savingenough,
  Author = {Skinner, Jonathan},
  Title = {Are You Sure You're Saving Enough for Retirement?},
  Journal = {Journal of Economic Perspectives},
  Volume = {21},
  Number = {3},
  Year = {2007},
  Month = {September},
  Pages = {59--80},
  DOI = {10.1257/jep.21.3.59}
}

@article{song2019compound,
  title={Financial Illiteracy and Pension Contributions: A Field Experiment on Compound Interest in China},
  author={Song, Changcheng},
  journal={The Review of Financial Studies},
  volume={33},
  number={2},
  pages={916--949},
  year={2020},
  publisher={Oxford University Press}
}

@article{stango2009exponential,
  title={Exponential Growth Bias and Household Finance},
  author={Stango, Victor and Zinman, Jonathan},
  journal={The Journal of Finance},
  volume={64},
  number={6},
  pages={2807--2849},
  year={2009},
  publisher={Wiley Online Library}
}

@article{wager2018forest,
  title={Estimation and Inference of Heterogeneous Treatment Effects using Random Forests},
  author={Wager, Stefan and Athey, Susan},
  journal={Journal of the American Statistical Association},
  volume={113},
  number={523},
  pages={1228--1242},
  year={2018},
  publisher={Taylor \& Francis}
}

@article{laibson2018private,
  title={Private Paternalism, the Commitment Puzzle, and Model-Free Equilibrium},
  author={Laibson, David},
  journal={AEA Papers and Proceedings},
  volume={108},
  pages={1--21},
  year={2018}
}

@unpublished{bodie1989pensions,
  title={Pensions as Retirement Income Insurance},
  author={Bodie, Zvi},
  year={1989},
  note={National Bureau of Economic Research}
}

@article{cravens2006employees,
  title={Employees: The Key Link to Corporate Reputation Management},
  author={Cravens, Karen S and Oliver, Elizabeth Goad},
  journal={Business Horizons},
  volume={49},
  number={4},
  pages={293--302},
  year={2006},
  publisher={Elsevier}
}

@article{eccles2007reputation,
  title={Reputation and Its Risks},
  author={Eccles, Robert G and Newquist, Scott C and Schatz, Roland},
  journal={Harvard Business Review},
  volume={85},
  number={2},
  pages={104},
  year={2007}
}

@misc{deloitte2014reputation,
    author={Deloitte},
    year={2014},
    title={2014 Global Survey on Reputation Risk},
    url={https://www2.deloitte.com/content/dam/Deloitte/pl/Documents/Reports/pl_Reputation_Risk_survey_EN.pdf}
}

@misc{OECD2019pensions,
  author = {OECD},
  title = {Pensions at a Glance},
  year = {2019},
  note = {data retrieved from OECD.Stat, 
          \url{https://stats.oecd.org/viewhtml.aspx?datasetcode=PAG&lang=en#}},
}

@misc{oldmutual2018,
  author = {{Old Mutual}},
  title = {Savings and Investment Monitor},
  year = {2018},
  note = {report available at \url{https://www.oldmutual.co.za/docs/default-source/personal-solutions/financial-planning/savings-and-monitor/2018-results-sim.pdf}},
}

@article{ambuehl2021motivates,
  title={What Motivates Paternalism? An Experimental Study},
  author={Ambuehl, Sandro and Bernheim, B Douglas and Ockenfels, Axel},
  journal={American Economic Review},
  volume={111},
  number={3},
  pages={787--830},
  year={2021}
}
}
\clearpage

\appendix

\begin{center}
    \LARGE{\textbf{Online Appendix}}
\end{center}

\thispagestyle{empty}

\setcounter{table}{0}
\renewcommand{\thetable}{A\arabic{table}}
\section{Descriptive Statistics}

\subsection{Employer-Provided Retirement Funds in South Africa}

In Table \ref{table:employer_fund}, using survey data from the National Income Dynamic Study (NIDS), I calculate the share of employed South Africans who have retirement contributions deducted from their salary by their employer. As expected, the share of employees with employer-provided retirement funds increases with salary. Only 14 percent of employees earning less than R 3 600 (250 USD) have retirement fund deductions whereas 66 percent of employees earning more than R48 000 per month (3 330 USD) have retirement fund deductions.

\begin{table}[ht]
\caption{Retirement fund offered by employer}
\label{table:employer_fund}
\small
\begin{threeparttable}
\begin{tabulary}{\textwidth}{cCCC}
\toprule
Monthly wage	&	Share of employed population	&	Retirement fund deducted from salary 	&	95\% confidence interval	\\
\midrule
$<$ R 3 600	&	26.91	&	14.17	&	8.89 - 19.47	\\
R 3 600 - R 8 000	&	32.85	&	29.37	&	25.48 - 33.26	\\
R 8 000 - R 15 000	&	17.53	&	46.84	&	42.92 - 50.77	\\
R 15 000 - R 22 000	&	8.10	&	60.25	&	53.89 - 66.61	\\
R 22 000 - R 48 000	&	9.55	&	63.62	&	58.23 - 69.01	\\
$>$ R 48 000	&	5.05	&	66.59	&	57.47 - 75.71 \\
\bottomrule
\end{tabulary}
\begin{tablenotes}
\item Source: Calculated from the National Income Dynamic Study (NIDS) Wave 5, which was conducted in 2017. The survey is nationally representative so the share shown in the table is a sample estimate of the share in the population. 
\end{tablenotes}
\end{threeparttable}
\end{table}

The experiment is conducted with employees of an asset management company. Although we don't have individual level salary data for each employee, industry standards suggest that most employees at the company will earn more than R 15 000 per month, corresponding to the bottom three rows of Table \ref{table:employer_fund}.

\subsection{Employees included in the experiment}

Table \ref{table:desc_sample} provides additional descriptive statistics of the 775 employees included in the experiment. 

\begin{table}[ht]
\centering
\caption{Descriptive statistics of employees included in the experiment}
\label{table:desc_sample}
\small
\begin{tabulary}{\textwidth}{lCCcc}
\toprule
	&	Mean	&	Standard deviation	&	Minimum	&	Maximum	\\
\midrule
Age (years) 	&	33.71	&	8.05	&	21	&	59	\\
Tenure (years) 	&	5.16	&	4.77	&	0	&	34	\\
Male 	&	0.47	&	0.5	&	0	&	1	\\
Previously disadvantaged	&	0.71	&	0.45	&	0	&	1	\\
\bottomrule
\end{tabulary}
\end{table}

\section{Email Treatment}\label{appendix:email}

The calculator was sent to the treated employees with the following email.

\

{\fontfamily{phv}\selectfont
\noindent \textbf{Are you on track for 75\%?}

\

\noindent Dear $<$Employee$>$

\

\noindent We all want to retire with enough. While there is no easy answer to what ‘enough’ is, a well-researched rule of thumb is that \textbf{a retirement income equal to 75\% of your final salary} will allow you to live comfortably in retirement. This figure accounts for the adjustments many people make as they grow older, for example, lower housing and higher medical costs.

\

\noindent Use our retirement income calculator to see how much your monthly salary income in retirement could be and follow the instructions on the webpage to make a change.

\

\noindent \underline{Use the calculator to see if you are on track.}
}


\setcounter{table}{0}
\renewcommand{\thetable}{C\arabic{table}}
\section{Checking Balance between Treatment and Control Groups}
\label{appendix:balance}


Table \ref{table:balance} provides summary statistics to check the balance between the treatment groups and the control groups. Since I stratified by age, gender and Apartheid-era classification, these characteristics will be balanced by design. The remaining measures which are likely to be correlated with the outcomes of interest include the pre-experiment contribution rate and tenure.

\begin{table}[ht]
\centering
\caption{Checking balance on observable characteristics}
\label{table:balance}
\small
\begin{tabular}{F{6cm} G{1.8cm} G{1.8cm} G{1.8cm}}
\toprule
	&	Control	&	Email	&	Email and phone	\\
\midrule
Observations	&	388	&	193	&	194	\\
\addlinespace
Pre-experiment contribution rate	&		&		&		\\
\ \    Mean	&	8.944	&	8.963	&	9.071	\\
\ \    Standard deviation	&	3.184	&	3.165	&	3.626	\\
\addlinespace
Tenure (years)	&		&		&		\\
\ \    Mean	&	5.838	&	5.566	&	5.638	\\
 \ \   Standard deviation	&	4.669	&	4.881	&	4.734	\\
\bottomrule
\end{tabular}
\end{table}

Since I stratified by whether the employee was saving at the minimum rate, the pre-experiment contribution rate could only differ between the treatment groups and the control groups if most of the employees contributing close to the maximum contribution rate happened to be randomized into one of the groups. This was not the case and the mean pre-experiment contribution rates are similar across groups. An F-test fails to reject the hypothesis of equal means across groups (p-value 0.88) and a t-test with the treatment groups combined also fails to reject the hypothesis of equal means (p-value 0.64).

For tenure, the means are also very similar across groups, possibly aided by stratifying by age categories. Again, an F-test fails to reject the hypothesis of equal means across groups (p-value 0.62) and a t-test with the treatment groups combined also fails to reject the hypothesis of equal means (p-value 0.33).

\setcounter{table}{0}
\renewcommand{\thetable}{D\arabic{table}}
\section{Contribution Rate Formula}\label{appendix:formula}

This appendix shows how to calculate the required contribution rate for a goal retirement income and illustrates why retirement contribution rates of 15 percent and higher are often recommended.

To calculate the required contribution rate for a goal retirement income, we must complete an annuity calculation. (The retirement calculator also completed the same type of calculation.) Let $R$ be the total accumulated retirement savings at retirement, $s$ the annual salary, $c$ the contribution rate, $r$ the investment return and $n$ the number of years until retirement. For simplicity, assume the employee has not started saving yet and the inflation rate is zero.
\begin{align*}
    R  &= sc \dfrac{(1+r)^n-1}{r}  
\end{align*}
Let $d$ be the sustainable drawdown rate (the percentage of the employee's retirement saving he withdraws each year during retirement) and $p$ the goal replacement rate (the ratio of retirement income to salary).
\begin{align*}
    ps = dR \\
    R = \dfrac{ps}{d}
\end{align*}
Substituting for $R$ and solving for $c$,
\begin{align*}
    \dfrac{ps}{d} &= sc \dfrac{(1+r)^n-1}{r}  \\
    c &= \dfrac{pr}{d[(1+r)^n-1]}
\end{align*}
As the formula highlights, the required contribution rate $c$ is sensitive to the assumptions of investment returns $r$ and the number of years of contributing $n$.

In Table \ref{table:contr_rates}, I show the required contribution rates for a range of starting and retirement ages. I assume a drawdown rate $d$ of 4 percent, a replacement rate $p$ of 75 percent and an investment return $r$ of 5 percent (which corresponds to the upper limit of the default investment return used by the calculator). The required contribution rates shown in the table highlight that unless an employee plans to save for more than 40 years or take a more optimistic view of investment returns, the required contribution rate is at least 15 percent. 

\begin{table}[ht]
    \centering
    \begin{tabular}{ccccccc}
     & & \multicolumn{5}{c}{Retirement age} \\
&	&   55	    &	60	    &	65	    &	70	    &	75	\\
\multirow{4}{*}{\rotatebox[origin=c]{90}{Start age}}
& 25	&	28.2\%	&	20.8\%	&	15.5\%	&	11.7\%	&	9.0\%	\\
& 30	&	39.3\%	&	28.2\%	&	20.8\%	&	15.5\%	&	11.7\%	\\
& 35	&  56.7\%	&	39.3\%	&	28.2\%	&	20.8\%	&	15.5\%	\\
& 40	& 	86.9\%	&	56.7\%	&	39.3\%	&	28.2\%	&	20.8\%
    \end{tabular}
    \caption{Required contribution rates to reach a replacement rate of 75 percent}
    \label{table:contr_rates}
\end{table}


\section{Heterogeneous Treatment Effects}
\label{sec:het}
\setcounter{figure}{0}
\renewcommand{\thefigure}{E\arabic{figure}}
\setcounter{table}{0}
\renewcommand{\thetable}{E\arabic{table}}

\subsection{Methodology}

Certain groups of employees may be more responsive to the calculator than others. I stratified the sample to detect heterogeneous effects. I use the following regression specification to detect heterogeneous treatment effects:
\begin{align}
Rate_i = \beta_0 + \beta_1 Treatment_i + \beta_2 Group_i + \beta_3 Treatment_i * Group_i  + \eta_i
\end{align}
$Treatment_i$ is the same treatment variable defined above and $Group_i$ is an indicator for group membership. 

I registered analysis of groups along three dimensions: gender, age and Apartheid-era racial categories. Since women live longer than men on average, the calculator recommends a higher savings rate for women. Older employees may respond differently to younger employees as retirement is a more immediate concern for older employees. Finally, employees who were previously disadvantaged by the Apartheid system may live with greater pressure to support family and have less freedom to save for retirement.

The registration of groups in the pre-analysis plan prevents the researcher from forming groups around outliers. However, this approach may lead to important groups being overlooked, especially if the groups are defined by combinations of observable characteristics. For example, men over the age of 35 may have much larger treatment effects than women over the age of 35. If we study heterogeneous treatment effects by age and gender separately, we may not detect this difference. 

To solve the problem of detecting possibly complex dimensions of heterogeneity while still retaining the transparency of a pre-analysis plan, I included \possessivecite{wager2018forest} causal forest method in the pre-analysis plan. Instead of trying to specify all possible dimensions of heterogeneity in advance, causal forests use a machine-learning algorithm to investigate the presence of heterogeneous effects.

Causal forests average over many causal trees. A causal tree represents a grouping of the observations. Figure \ref{fig:tree} provides an example of a causal tree. We start at the top of the tree with all the treatment and control observations. In this example, we start by splitting the observations on age. Employees younger than 35 are split into a group. If the group is not split further, we reach a terminal group, called a leaf. For employees older than 35, we split again by gender. This tree has three leaves, which gives us three groups. The tree is called \emph{causal} because we can estimate the conditional average treatment effect within each leaf. For example, in the leaf of employees younger than 35, we can calculate the average treatment effect using only the observations which fall within that leaf.

\begin{figure}
    \centering
    \includegraphics[height=8cm]{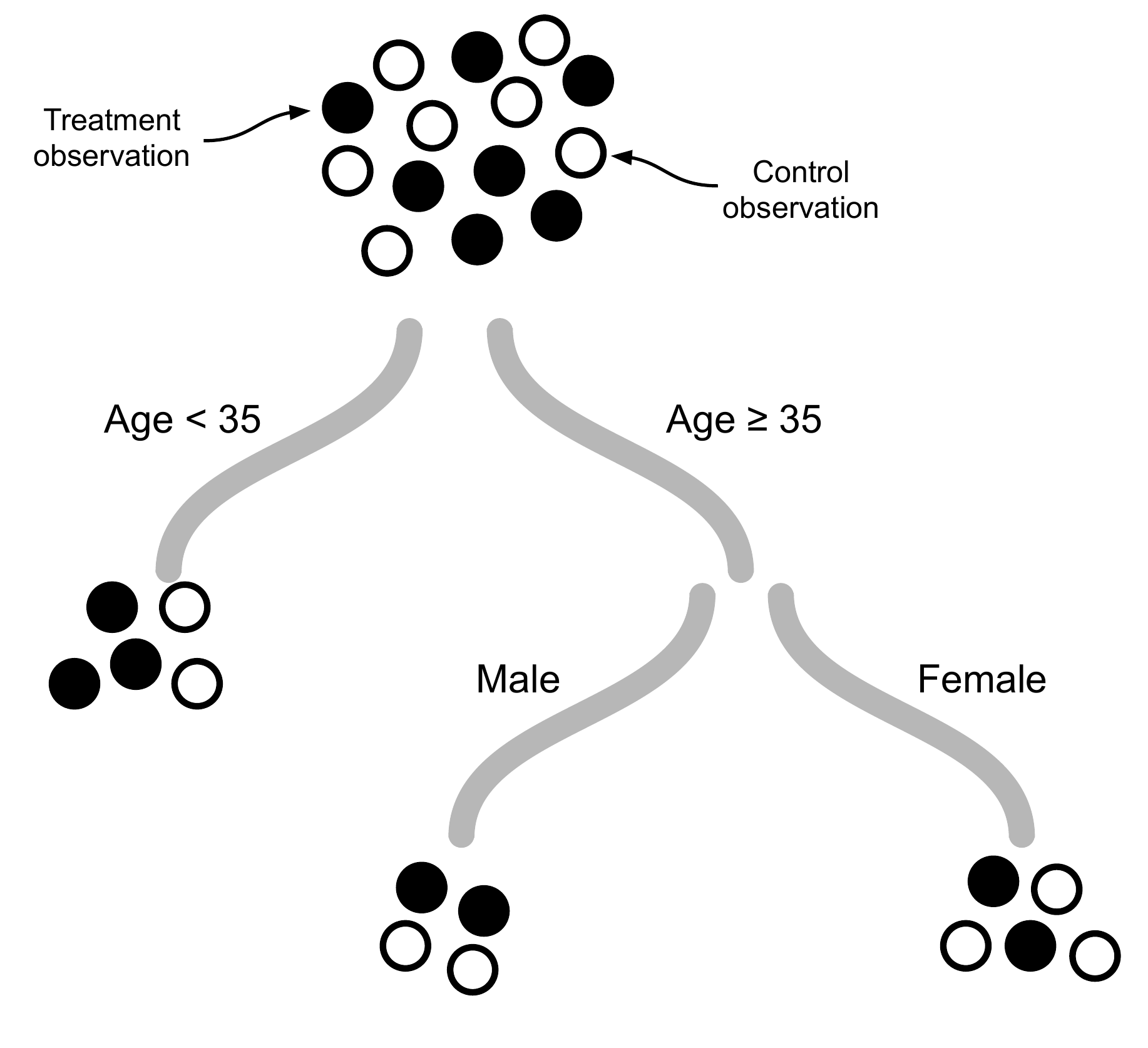}
    \caption{Example of a causal tree}
    \label{fig:tree}
\end{figure}

The algorithm forms a causal tree by choosing the groups which maximize the difference in treatment effects between groups. For each split, the algorithm calculates the estimated treatment effect in each leaf and chooses the split which maximizes the variance in estimated treatment effects across leaves, but with a penalty for the within leaf variance in treatment effects. The algorithm stops splitting if the variance of treatment effects across leaves cannot be increased or if there are too few treatment or control observations within the leaf.

Many causal trees are averaged to form the causal forest. To prevent correlation between trees, the algorithm uses a random subsample of observations and a random subsample of observable characteristics to form each tree. This process of repeated subsampling of both the observations and observables is the essence of the random forests algorithm, a popular machine-learning algorithm on which causal forests is based.

Causal forests have an ``honesty property" to prevent bias from outliers. The algorithm decides how to split a tree with one subsample, but estimates the conditional average treatment effect with a held-out subsample. Since causal forests average over many trees, all the observations are eventually used in both the splitting subsample and the held-out estimation subsample. 

I use the \texttt{grf} R package to estimate the causal forest.\footnote{Visit \href{https://github.com/grf-labs/grf}{github.com/grf-labs/grf} for details on how to use the \texttt{grf} package.} I build 2 000 trees. For each tree, the algorithm takes a 25 percent subsample to build the tree and a 25 percent subsample to estimate the conditional treatment effects in each leaf. Leaves with less than 5 treatment or control observations from the estimation subsample are not split further. I include age categories, gender, an indicator for classified as ``white" during the Apartheid regime and the pre-experiment contribution rate as inputs for the regression trees.

The final output of the causal forest provides a model of the conditional average treatment effects. We can predict each individual's conditional average treatment effect using the model. An individual is passed down every tree in the forest based on his or her observable characteristics. For each tree, the individual will reach a leaf, or terminal group. The individual's conditional average treatment effect estimate is the average over the estimates in each leaf he or she reaches.

\subsection{Results}

I investigate the presence of heterogeneous impacts by gender, Apartheid-era racial classifications and age. I registered these dimensions of heterogeneity in the pre-analysis plan. The results are shown in separate tables below.

\begin{table}[ht!]
\caption{Effect by Gender}
\begin{center}
\scalebox{0.85}{
\begin{threeparttable}
\begin{tabular}{l c c}
\toprule
 & Contribution rate & Change in contribution rate \\
\midrule
Email only             & $-0.204$           & $-0.054$           \\
                       & $ [-0.851; 0.442]$ & $ [-0.217; 0.110]$ \\
Phone and email        & $0.184$            & $-0.007$           \\
                       & $ [-0.669; 1.036]$ & $ [-0.206; 0.192]$ \\
Male                   & $-0.006$           & $-0.144$           \\
                       & $ [-0.662; 0.651]$ & $ [-0.352; 0.065]$ \\
Email only * Male      & $0.709$            & $0.258$            \\
                       & $ [-0.441; 1.859]$ & $ [-0.030; 0.545]$ \\
Phone and email * Male & $0.002$            & $0.126$            \\
                       & $ [-1.210; 1.214]$ & $ [-0.196; 0.449]$ \\
\midrule
Control mean           & $9.010$            & $0.066$            \\
R$^2$                  & $0.003$            & $0.005$            \\
Num. obs.              & $765$              & $765$              \\
RMSE                   & $3.344$            & $0.912$            \\
\bottomrule
\end{tabular}
\begin{tablenotes}[flushleft]
\small{\item All specifications use robust standard errors. ATE stands for Average Treatment Effect.}
\end{tablenotes}
\end{threeparttable}
}
\label{table:het_gender}
\end{center}
\end{table}

\begin{table}[ht!]
\caption{Effect by Apartheid-era categories}
\begin{center}
\scalebox{0.85}{
\begin{threeparttable}
\begin{tabular}{l c c}
\toprule
 & Contribution rate & Change in contribution rate \\
\midrule
Email only                   & $-0.002$           & $0.038$            \\
                             & $ [-0.580; 0.576]$ & $ [-0.111; 0.186]$ \\
Phone and email              & $0.014$            & $0.017$            \\
                             & $ [-0.567; 0.595]$ & $ [-0.162; 0.197]$ \\
Privileged                   & $1.272^{*}$        & $-0.121$           \\
                             & $ [ 0.520; 2.024]$ & $ [-0.385; 0.143]$ \\
Email only * Privileged      & $0.437$            & $0.110$            \\
                             & $ [-0.970; 1.844]$ & $ [-0.250; 0.470]$ \\
Phone and email * Privileged & $0.531$            & $0.126$            \\
                             & $ [-1.037; 2.100]$ & $ [-0.252; 0.504]$ \\
\midrule
Control mean                 & $9.010$            & $0.066$            \\
R$^2$                        & $0.044$            & $0.003$            \\
Num. obs.                    & $765$              & $765$              \\
RMSE                         & $3.275$            & $0.913$            \\
\bottomrule
\end{tabular}
\begin{tablenotes}[flushleft]
\small{\item All specifications use robust standard errors. ATE stands for Average Treatment Effect.}
\end{tablenotes}
\end{threeparttable}
}
\label{table:het_aparthied}
\end{center}
\end{table}

\begin{table}[ht]
\caption{Effect by Age Categories}
\begin{center}
\scalebox{0.85}{
\begin{threeparttable}
\begin{tabular}{l c c}
\toprule
 & Contribution rate & Change in contribution rate \\
\midrule
Email only              & $-0.144$           & $-0.061$           \\
                        & $ [-1.033; 0.744]$ & $ [-0.466; 0.344]$ \\
Phone and email         & $-0.583$           & $-0.242$           \\
                        & $ [-1.290; 0.125]$ & $ [-0.650; 0.167]$ \\
27-32                   & $-0.403$           & $-0.190$           \\
                        & $ [-1.131; 0.325]$ & $ [-0.561; 0.182]$ \\
32-38                   & $0.280$            & $-0.135$           \\
                        & $ [-0.716; 1.275]$ & $ [-0.492; 0.221]$ \\
38+                     & $0.225$            & $-0.208$           \\
                        & $ [-0.696; 1.146]$ & $ [-0.563; 0.147]$ \\
Email only * 27-32      & $0.776$            & $0.241$            \\
                        & $ [-0.700; 2.252]$ & $ [-0.231; 0.713]$ \\
Email only * 32-38      & $0.188$            & $0.153$            \\
                        & $ [-1.320; 1.696]$ & $ [-0.311; 0.617]$ \\
Email only * 38+        & $0.199$            & $0.147$            \\
                        & $ [-1.359; 1.757]$ & $ [-0.289; 0.582]$ \\
Phone and email * 27-32 & $1.399$        & $0.310$            \\
                        & $ [ 0.034; 2.763]$ & $ [-0.195; 0.816]$ \\
Phone and email * 32-38 & $0.095$            & $0.409$            \\
                        & $ [-1.210; 1.399]$ & $ [-0.063; 0.881]$ \\
Phone and email * 38+   & $1.658$            & $0.505$        \\
                        & $ [-0.169; 3.486]$ & $ [ 0.016; 0.994]$ \\
\midrule
Control mean            & $9.010$            & $0.066$            \\
R$^2$                   & $0.016$            & $0.010$            \\
Num. obs.               & $765$              & $765$              \\
RMSE                    & $3.336$            & $0.914$            \\
\bottomrule
\end{tabular}
\begin{tablenotes}[flushleft]
\small{\item All specifications use robust standard errors. ATE stands for Average Treatment Effect.}
\end{tablenotes}
\end{threeparttable}
}
\label{table:het_age}
\end{center}
\end{table}

I cannot pre-specify all possible dimensions of heterogeneity and certain subgroups with large treatment effects may go undetected, especially if the subgroups are defined by complex interactions of the measured characteristics. To tackle this problem, I use a machine learning algorithm called causal forests developed by \citet{wager2018forest}. The algorithm provides individual-level estimates of conditional treatment effects, akin to estimates provided by a nearest neighbor matching approach. 

I plot a histogram of the causal forest estimates of the conditional treatment effects for the 765 employees in Figure \ref{fig:forest_hist}. I combine the email and email and phone treatments into a single binary treatment indicator for having received an email about the calculator. The conditional treatment effects range between -0.3 and 0.3. Since the conditional treatment effects are clustered in a tight range around the average treatment effect, we can conclude that the treatment effect is similar across subgroups.

\begin{figure}[ht]
\centering
\includegraphics[width=\textwidth]{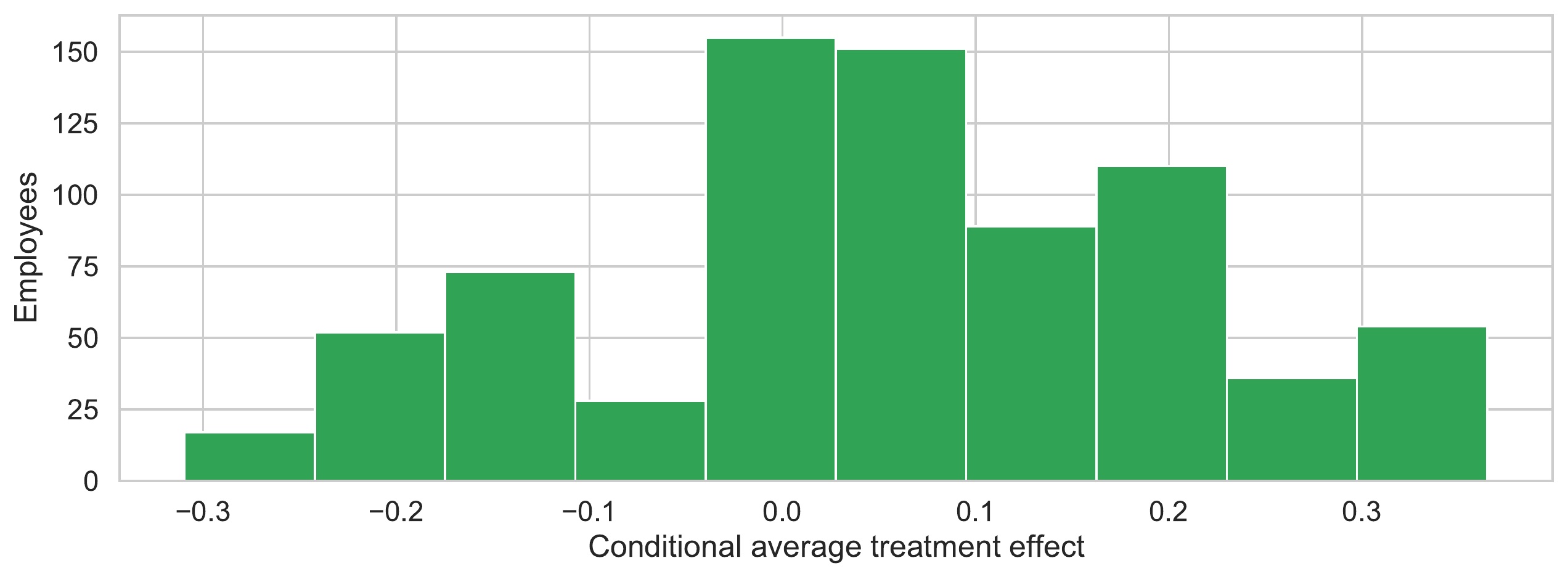}
\caption{Histogram of Conditional Treatment Effects}
\label{fig:forest_hist}
\end{figure}




\section{Examples of Company Retirement Policies}
\label{sec:example_policies}
\setcounter{figure}{0}
\renewcommand{\thefigure}{F\arabic{figure}}
\setcounter{table}{0}
\renewcommand{\thetable}{F\arabic{table}}

The financial services company in this study set a mandatory minimum contribution rate of 7.5 percent towards the retirement account. In this section, I provide examples to show that this policy is not unique to this company.

I collected a list of 2000 companies that provide goods or services to South African customers from the HelloPeter complaints website. All companies received at least 10 reviews on the website. I then collected contact details for the companies with the assistance of workers on Amazon Mechanical Turk. I sent the following email to each company:

\ 

{\fontfamily{phv}\selectfont

\noindent \textbf{Survey on corporate retirement policies}

\noindent I am a research fellow at UNSW Sydney, and I am studying corporate retirement savings policies. Could you please answer the following brief questions about [Company Name]?

\

\noindent 1. Does [Company Name] offer permanent employees a retirement savings policy (such as a pension fund, provident fund or retirement annuity)?

\noindent If yes:

\noindent 2.1. Is it mandatory for the employee to join the retirement policy?

\noindent 2.2. What is the minimum percentage of salary the employee can contribute to the retirement policy?

\ 
}

In Table \ref{table:company_policies}, I provide the details of the companies who responded to the survey. Half of the companies who responded do impose a minimum contribution rate on employees. The minimum varies from 5.5 to 15.5 percent of the employee's salary. However, there are also examples of companies that do not impose minimums. 

\begin{table}[ht]
\caption{Company Retirement Policies}
\begin{center}
\scalebox{0.75}{
\begin{threeparttable}
\begin{tabular}{l I I I}
\toprule
Industry	&	Employees	&	Mandatory retirement savings?  	& Minimum contribution		\\
\midrule
Telecommunications	&	$>$ 5 000	&	Yes	&	12	\\
Financial Services	&	$>$ 5 000	&	Yes	&	9.5	\\
Financial Services	&	$>$ 5 000	&	Yes	&	10	\\
Financial Services	&	$>$ 5 000	&	Yes	&	14	\\
Clothing Supplier	&	760	&	Yes	&	10.5	\\
Telecommunications	&	700	&	Yes	&	7	\\
Real Estate	&	320	&	Yes	&	7	\\
Security	&	230	&	Yes	&	6	\\
Financial Services	&	210	&	No	&		\\
Printing	&	170	&	No	&	-	\\
Education	&	140	&	No	&	-	\\
Automotive	&	140	&	Yes	&	10	\\
Entertainment	&	130	&	No	&		\\
Retail	&	90	&	No	&	-	\\
Consumer Electronics	&	80	&	No	&	15	\\
Accommodation	&	50	&	No	&	-	\\
Logistics	&	40	&	Yes	&	7.5	\\
Real Estate	&	40	&	No	&	-	\\
Real Estate	&	40	&	No	&	-	\\
Real Estate	&	30	&	No	&	-	\\
Real Estate	&	20	&	Yes	&	5	\\
Consumer Services	&	20	&	Yes	&	10	\\
Restaurant	&	10	&	Yes	&	5.5	\\
Financial Services	&	10	&	No	&	-	\\
Education	&	10	&	No	&	-	\\
Automotive	&	10	&	Yes	&	15.5	\\
Retail	&	10	&	No	&	- \\
\bottomrule
\end{tabular}
\begin{tablenotes}[flushleft]
\small{\item The minimum contribution is measured as a percentage of the employee's gross salary. Company names are omitted to preserve the company's privacy.}
\end{tablenotes}
\end{threeparttable}
}
\label{table:company_policies}
\end{center}
\end{table}


\end{document}